\documentclass[11pt,twoside]{article}
\usepackage{url}
\usepackage{xspace}
\usepackage{latexsym}
\usepackage{epsfig}
\usepackage{fancybox}
\usepackage{amssymb}
\usepackage{amsmath}
\usepackage{amsfonts}
\usepackage{rotating}
\usepackage{booktabs}
\usepackage{amsthm,stmaryrd}
\usepackage{indentfirst}
\usepackage{color}
\usepackage{geometry}
\usepackage{tikz}
\usepackage{graphicx}
\usepackage{wrapfig}
\usepackage{fancyhdr}
\usepackage{fullpage}
\usepackage{hyperref}
\hypersetup{
    colorlinks=true,
    citecolor=blue,
    linkcolor=red,
    filecolor=magenta,
    urlcolor=cyan,
}

\newtheorem{thm}{Theorem}[section]
\newtheorem{theorem}[thm]{Theorem}

\newtheorem{cor}[thm]{Corollary}
\newtheorem{lemma}[thm]{Lemma}
\newtheorem{prop}[thm]{Proposition}
\newtheorem{definition}[thm]{Definition}

\newtheorem{remark}[thm]{Remark}
\newtheorem{example}[thm]{Example}

\newtheorem{problem}[thm]{Problem}

\makeatletter \@addtoreset{equation}{section}

\newcommand{\bN} { {\mathbb{N}}}
\newcommand{\bC} { {\mathbb{C}}}
\newcommand{\bQ} { {\mathbb{Q}}}
\newcommand{\bZ} { {\mathbb{Z}}}

\title{Stability Problems in Symbolic Integration}

\author{Shaoshi Chen
  }
\date{\small
 KLMM, Academy of Mathematics and Systems Science,\\ Chinese Academy of Sciences, Beijing, 100190,
  China\\
  {\sf schen@amss.ac.cn}
}

\begin{document}
\maketitle

\begin{abstract}
This paper aims to initialize a dynamical aspect of symbolic integration by studying stability problems in differential fields. We present some basic properties of stable elementary functions and D-finite power series that enable us to characterize
three special families of stable elementary functions involving rational functions, logarithmic functions, and exponential functions. Some problems for future studies are proposed towards deeper dynamical studies in differential and difference algebra.
\end{abstract}

\section{Introduction}\label{SECT:intro}
In the proofs of the irrationality and transcendence of $e$ and $\pi$,
various kinds of definite integrals such as
\[\int_{-1}^1 (1-x^2)^n \cos(xz)\, dx \quad \text{and} \quad \int_0^{\pi} \frac{x^n(a-bx)^n}{n!} \sin(x)\, dx\]
are used~\cite{Niven1956, Natarajan2020}. In the process of deriving linear recurrences in $n$ for these integrals, one may realize that the integrands are so nice that their shape is stable under indefinite integration.
Some typical such kind of nice functions are polynomials, radicals and basic transcendental elementary functions.
One may be curious about whether these are the only possible functions that have this feature. This motivates our dynamical thinking in symbolic integration.

The problem of integration in finite terms or in closed form is one of the oldest problems in calculus. Since the initial work of Liouville from 1832 to 1841,  different systematical approaches have been developed to study this problem~\cite{Ostrowski1946, Ritt1948, Hardy1916, Khovanskii2014}.
A comprehensive historical study of Liouville's work on integration in finite terms is given in~\cite{LutzenBook} and also some friendly introductory notes are~\cite{Kasper1980, Elena1994}. After the birth of differential algebra, the problem is formulated in a pure algebraic fashion
and then solved by Risch~\cite{Risch1969, Risch1970} with further developments in symbolic integration~\cite{Slagle1961, Moses1968, Rothsetin1976,  Davenport1981, Trager1984, Bronstein1990, Gerhard2001, BronsteinBook, Raab2012}. The aim of symbolic integration is developing practical algorithms and softwares for solving the integration problem of elementary functions and other more general special functions from mathematical physics. The standard references on symbolic integration are firstly Bronstein's book~\cite{BronsteinBook} and some chapters in~\cite{MCA2003, ACA1992}. Recently, Raab gave an informative survey~\cite{Raab2013} on the Risch algorithm and its recent developments.
Another significant current trend is the arithmetic studies of elementary integration initialized by Masser an Zannier in~\cite{Zannier2014, Davenport2016, Masser2017, MZ2020}.

Self-maps on structured sets are ubiquitous in mathematics. The theory of dynamical systems is the mathematics of self-maps. So
the dynamical way of thinking has inspired many interdisciplinary areas in mathematics such as arithmetic dynamics in number theory~\cite{SilvermanBook}
and complex dynamics in analysis~\cite{Beardon1991}.
The fundamental objects in differential algebra are  differential fields and their extensions. The indefinite integration problem of elementary functions can be formulated in terms of differential fields.
In this paper, we view
differential fields as dynamical systems where derivations
play the role of self-maps.  Through the dynamical lens, we will see
some new landscape of symbolic integration and also understand better why the integrals from the beginning are so nice.

The remainder of this paper is organized as follows. We recall some basic terminologies in dynamical systems and differential algebra and then define
stability problems in Section~\ref{SECT:stable}. After bridging the connection, we explore the basic properties of stable elementary
functions and characterize three families of special stable
elementary functions in Section~\ref{SECT:elefun}. In Section~\ref{SECT:dfinite}, we study the stability problem on D-finite powers series and then conclude our paper in ~\ref{SECT:conc} by proposing
some problems for future research.

\section{Stability in differential fields}\label{SECT:stable}
We first bridge the connection between dynamical systems and differential algebra.  A (discrete) \emph{dynamical system} is a pair $(A, \phi)$ with $A$ being a set and $\phi: A\rightarrow A$ being a self-map on $A$. We recall the definition of four special subsets that are crucial for understanding a dynamical system.

\begin{definition} \label{DEF:stable}
Let $(A, \phi)$ be a dynamical system and $a\in A$. Then we say that

\begin{itemize}
\item[(1)] the element $a$ is a \emph{fixed} point of $\phi$ if $\phi(a)=a$.  The set of all fixed points is denoted by $\text{Fix}(\phi, A)$;
\item[(2)] the element $a$ is a \emph{periodic} point of $\phi$ if $\phi^n(a)=a$ for some positive $n\in \bN$. The set of all periodic points is denoted by $\text{Per}(\phi, A)$;
\item[(3)] the element $a$ is \emph{stable} in the system $(A, \phi)$ if there exists a sequence $\{a_i\}_{i\geq 0}$ in $A$ such that $a_0 = a$ and $\phi(a_{i+1}) = a_i$ for all $i\in \bN$. The set of all stable elements is denoted by $\text{Stab}(\phi, A)$;
\item[(4)] the element $a$ is \emph{attractive} in the system $(A, \phi)$ if for any $i\in \bN$, there exists
$a_i\in A$ such that $a = \phi^i(a_i)$. The set of all attractive elements is denoted by $\text{Attrac}(\phi, A)$ that is equal to $\bigcap_{i\in \bN}\phi^i(A)$.
\end{itemize}
\end{definition}

For any dynamical system $(A, \phi)$, we have the  inclusions:
\[\text{Fix}(\phi, A) \subseteq \text{Per}(\phi, A) \subseteq \text{Stab}(\phi, A) \subseteq \text{Attrac}(\phi, A).\]
It is not hard to see that the first three inclusions may be proper by definition. For the last one,  Godelle in~\cite{Godelle2010} presented a concrete example as below and also
discussed the question of deciding whether the inclusion $\text{Stab}(\phi, A) \subseteq \text{Attrac}(\phi, A)$ is indeed an equality in various settings.
\begin{example}[Godelle's Example]\label{EXAM:godelle}
Let $A = \{(i, j)\in \bZ^2\mid 0\leq j \leq \max(i-1, 0)\}$ and
let $\phi: A\rightarrow A$ be defined by $\phi(i, j) = (i, j-1)$ for positive $j$, and $\phi(i, 0) = (\min(i-1, 0), 0)$. Then the subset $\text{Stab}(\phi, A)$
is empty, but the subset $\text{Attrac}(\phi, A)$ is $\{(i, 0)\mid i\leq 0\}$.
\end{example}

For the convenience of later usage, we summarize the results on stable and attractive subsets
in~\cite[Sections 2-3]{Godelle2010} as follows.
\begin{theorem}[Godelle's Theorem]\label{THM:godelle}
Let $A$ be a set and $\phi: A\rightarrow A$ be a self-map on $A$. Then
\begin{itemize}
\item[$(i)$] The stable subset $\text{Stab}(\phi, A)$ is the largest one among all subsets of $A$ satisfying the property  $\phi(B) = B$ for $B \subseteq A$.
\item[$(ii)$] If $\phi$ is either injective on some $\phi^n$ or surjective, then
$\text{Stab}(\phi, A) = \text{Attrac}(\phi, A)$.
\item[$(iii)$] If $A$ is an infinite-dimensional vector space over a field $k$, and $\phi$ is a linear map such that $\dim_k(\text{Ker}(\phi))=1$. Then $\text{Stab}(\phi, A) = \text{Attrac}(\phi, A)$.
\end{itemize}
\end{theorem}


Throughout this paper, we assume that all fields are of characteristic zero. For any field $K$, we call an additive map $\delta: K\rightarrow K$
 a \emph{derivation} on $K$ if $\delta(fg) = g\delta(f) + f \delta(g)$
for all $f, g\in K$. By induction, we have the general Leibniz formula
\begin{equation}\label{EQ:leibniz}
\delta^n(fg) = \sum_{i=0}^n \binom{n}{i}\delta^i(f)\delta^{n-i}(g)
\end{equation}
for all $n\in \bN$ and $f, g\in K$.  We call the pair $(K, \delta)$ a \emph{differential field}.
An element $c\in K$ is called a \emph{constant} in $K$ if $\delta(c)=0$. The set of all constants in $K$ forms a subfield of $K$, denoted by $C_K$. Note that
$\delta$ is a $C_K$-linear map on $K$. An element $f\in K$
is said to be \emph{integrable} in $K$ if there exists $g\in K$ such that $f = \delta(g)$. By the $C_K$-linearity of the derivation $\delta$, we see that the set $\delta(K)$ of all integrable elements in $K$ forms a linear subspace of $K$ over $C_K$.  A differential field $(K_1, \delta_1)$ is called a \emph{differential extension} of the differential field $(K, \delta)$ if $K\subseteq K_1$
and $\delta_1{\mid}_{K} = \delta$. By abuse of notation, we use the same symbol for derivations on differential extensions.

The ring of linear differential operators over the differential field $(K, \delta)$ is denoted by $K\langle D \rangle$, in which we have the commutation rule
\[D\cdot f = f\cdot D + \delta(f)\quad \text{for all $f\in K$}. \]
One of basic properties of the ring $K\langle D \rangle$ is that it is a left Euclidean domain, in which we can define and effectively compute the greatest common right divisor (GCRD) and least common left multiple (LCLM) of polynomials if $K$ is a computable field. For more results on this ring, one can see~\cite{Ore1933, BronsteinPetkovsek1996, AbramovLeLi2005}.
The general Leibniz rule~\eqref{EQ:leibniz} now is translated into the
general commutation rule in $K\langle D \rangle$ as follows
\begin{equation}\label{EQ:drule1}
D^n\cdot f  = \sum_{i=0}^n \binom{n}{i}\delta^i(f)\cdot D^{n-i}
\end{equation}
for all $n\in \bN$ and $f\in K$. Similarly, we have the following formula
\begin{equation}\label{EQ:drule2}
f \cdot D^n  = \sum_{i=0}^n (-1)^i \binom{n}{i} D^{n-i} \cdot \delta^i(f)
\end{equation}
for all $n\in \bN$ and $f\in K$. The differential field $K$ can be viewed as
a $K\langle D \rangle$-module via the action
\[L(f) = \sum_{i=0}^d \ell_i \delta^i(f)  \]
for any $L = \sum_{i=0}^d \ell_i D^i\in K\langle D \rangle$ and $f\in K$.

We now view the differential field $(K, \delta)$
as a dynamical system and consider its stable and attractive subsets.

\begin{prop}\label{PROP:stable}
Let $(K, \delta)$ be a differential field. Then
\[\text{Stab}(\delta, K) = \text{Attrac}(\delta, K).\]
\end{prop}
\begin{proof}
Let $C_k$ be the constant subfield of $K$. Then $K$ can be viewed as a vector space over $C_K$. If there exists $f\in K$ such that $\delta(f)\neq 0$, then $f$ must be transcendental over $C_K$ by Lemma 3.3.2 in~\cite[p.\ 86]{BronsteinBook}. So the dimension $\dim_{C_K}(K)$ is either one or infinite. When
$\dim_{C_K}(K)$ is one, we have $K=C_K$ and then both  $\text{Stab}(\delta, K)$ and $\text{Attrac}(\delta, K)$ are equal to $\{0\}$. So the equality holds.
When $\dim_{C_K}(K)$ is infinite, we get the equality by Proposition 3.3 in~\cite{Godelle2010} since $\delta$ is a linear map on $K$ over $C_K$ and $\dim_{C_K}(\text{Ker}(\delta)) = 1$.
\end{proof}
So it is not needed to distinguish the stable elements and attractive elements in a
dynamical system arising from a differential field.  The central problem considered in this paper is as follows.

\begin{problem}[Stability Problem]\label{PROB:stable}
Given an element $f$ in a differential field $(K, \delta)$, decide
whether $f$ is stable in $(K, \delta)$ or not. More generally, describe the
structure of the stable subset $\text{Stab}(\delta, K)$.
\end{problem}

\begin{definition}\label{DEF:regular}
A differential field  $(K, \delta)$ is said to be~\emph{regular} if there exists $x\in K$ such that $\delta(x)=1$.
\end{definition}

\begin{example}\label{EXAM:regular}
Let $K$ be the field $\bC(x)$ of rational functions in $x$ over
the field of complex numbers $\bC$. If $\delta$ is the usual derivation $d/dx$, then
$(K, \delta)$ is regular since $\delta(x)=1$. If $\delta$ is the Eulerian derivation
$x\cdot d/dx$, then $(K, \delta)$ is not regular since $1/x$ has no anti-derivative in $\bC(x)$.
\end{example}

\begin{prop}\label{PROP:regular}
A differential field $(K, \delta)$ is regular if and only if all
constants are stable in $(K, \delta)$. Moreover, an element $f$ in a regular differential field $(K, \delta)$ is stable if and only if $\delta(f)$ is stable.
\end{prop}
\begin{proof}
Suppose that $(K, \delta)$ is regular. Then there exists $x\in K$ such that $\delta(x)=1$. Let $c$ be any constant in $K$.  Note that
\[c = \delta^i\left(\frac{c}{i!} x^i\right) \quad \text{for any $i\in \bN$}.\]
So $c\in \text{Attrac}(\delta, K) = \text{Stab}(\delta, K)$. The necessity is obvious since $1$ is always a constant.

For the second assertion, suppose that $f$ is stable in $K$. Then Theorem~\ref{THM:godelle}~$(i)$ says that the stability is preserved under the derivation. So $\delta(f)$ is stable.  Conversely, suppose that $\delta(f)$ is stable. Then there exists a sequence $\{g_i\}_{i\in \bN}$ in $K$ such that $g_0 = \delta(f)$ and $g_i = \delta(g_{i+1})$ for all $i\in \bN$. It is clear that all of the $g_i$'s are stable by definition. By the equality $\delta(f) = \delta(g_1)$, there exists $c\in C_K$ such that $f = g_1 + c$. Since $K$ is regular, $c$ is also stable in $K$ by the first assertion. Thus, $f$ is stable in $K$.
\end{proof}

\begin{remark}\label{REM:regular}
We point out that the regularity assumption in the above proposition is needed.
In the differential field $(\bC(x), x\cdot d/dx)$ which is not regular shown in Example~\ref{EXAM:regular}, any nonzero constant $c$ is not stable but its derivative is always stable.
 \end{remark}

\begin{lemma}\label{LEM:stable}
Let $(K, \delta)$ be a regular differential field with $\delta(x)=1$ and $f\in K$. Then
\begin{itemize}
\item[$(i)$] Let $n$ be a positive integer. Then $f = \delta^n(g)$ for some $g\in K$ if and only if for any $i$ with $0\leq i \leq n-1$, there exists $h_i\in K$ such that $x^i f = \delta(h_i)$.
\item[$(ii)$] $f$ is stable in $(K, \delta)$ if and only if
for all $i\in \bN$, $x^i f = \delta(g_i)$ for some $g_i\in K$.
\end{itemize}
\end{lemma}
\begin{proof} For showing the sufficiency of the first assertion, we suppose that $f = \delta^n(g)$ for some $g\in K$. Then the formula~\eqref{EQ:drule2} implies that for any $i$ with $0\leq i \leq n-1$, we have
\[x^i\cdot  D^n = \sum_{j=0}^n (-1)^j \binom{n}{j}D^{n-j}\cdot \delta^j(x^i). \]
Since $\delta^j(x^i) =0$ if $j>i$, we have $x^i\cdot  D^n = D\cdot L_i$ with
\[L_i = \sum_{j=0}^i (-1)^j j! \binom{n}{j} \binom{i}{j}D^{n-j-1}\cdot x^{i-j}.\]
Then $x^i f = x^i \delta^n(g) = (x^i\cdot D^n)(g) = \delta(h_i)$ with $h_i = L_i(g)\in K$.

We prove the necessity by induction on $n$.  The assertion holds obviously in the base case when $n=1$. We now assume that the assertion hold for $n< m$. Suppose that
there exist $h_0, \ldots, h_{m-1}\in K$ such that
\[f = \delta(h_0), \quad \ldots, \quad x^{m-1}f = \delta(h_{m-1}).\]
The goal is to show that $f = \delta^m(g)$ for some $g\in K$.
By the induction hypothesis, there exists $u\in K$ such that $f = \delta^{m-1}(u)$ using the first $m-1$ equalities $x^j f = \delta(h_j)$ for $j$ with $0\leq j \leq m-2$. By the formula~\eqref{EQ:drule2}, we have
\[ x^{m-1} \cdot D^{m-1}  = D \cdot L + (m-1)!\]
for some $L \in K\langle D \rangle $. We claim that we
can choose
\[g =\frac{1}{(m-1)!}(h_{m-1} - L(u)).  \]
Since $f = \delta^{m-1}(u)$ and $x^{m-1} f = \delta(h_{m-1})$, we have
\begin{align*}
D^{m-1} \cdot (x^{m-1} f) & =(D^{m-1} \cdot (x^{m-1} \cdot D^{m-1}))(u)\\
& = (D^{m-1}\cdot (DL + (m-1)!))(u) \\
& =  D^m(L(u)) + (m-1)!D^{m-1}(u)\\
& = \delta^m(L(u)) + (m-1)! f.
\end{align*}
Since $D^{m-1} \cdot (x^{m-1} f) = \delta^{m-1}(\delta(h_{m-1})) = \delta^m(h_{m-1})$, we get $\delta^m(L(u)) + (m-1)! f = \delta^m(h_{m-1})$, which implies the claim.

The second assertion follows immediately from the first one.
\end{proof}

\begin{theorem}\label{THM:stable}
Let $(K, \delta)$ be a regular differential field with $\delta(x)=1$. Then the stable subset $\text{Stab}(\delta, K)$ forms a $C_K[x]$-module and it is closed under differentiation.
\end{theorem}
\begin{proof}
Since the derivation $\delta$ is a linear map on $K$ over $C_K$, we have $\delta^i(K)$ is a linear subspace of $K$ for any $i\in \bN$. By Proposition~\ref{PROP:stable}, $\text{Stab}(\delta, K) = \cap_{i\in \bN} \delta^i(K)$ is also a linear subspace of $K$ over $C_K$. To show that $\text{Stab}(\delta, K)$ forms a $C_K[x]$-module, it suffices to prove that for any $f\in \text{Stab}(\delta, K)$ and $n\in \bN$, we have $x^n f \in \text{Stab}(\delta, K)$. By Lemma~\ref{LEM:stable}, an element $g\in K$ is stable if and only if for all $i\in \bN$, $x^i g = \delta(h_i)$ for some $h_i\in K$. Since $f$ is stable, we have for all $j\in \bN$, $x^j f = \delta(u_j)$ for some $u_j\in K$. This implies that
 for all $i, n\in \bN$, $x^i(x^n f) = \delta(v_{i, n})$ for some $v_{i, n}\in K$.
 So $x^nf$ is stable for any $n\in \bN$.  The closeness of $\text{Stab}(\delta, K)$  under differentiation follows from Theorem~\ref{THM:godelle}~$(i)$.
\end{proof}

We now recall the notion of orders of rational functions and their basic properties from~\cite[Chapter 4]{BronsteinBook}. Let $F$ be a field of characteristic zero and $F(t)$ be the field of rational functions in $t$ over $F$. Let $p\in F[t]$ be an irreducible polynomial and $f\in F(t)$. We can always write $f = p^m a/b$, where $m \in \bZ, a, b \in F[t]$ with $\gcd(a, b) = 1$ and $p\nmid ab$. We call the integer $m$ the \emph{order} of $f$ at $p$, denoted by $\nu_p(f)$. Conventionally, we set $\nu_p(0)= +\infty$. Let $\delta$ be a derivation on $F$ and we can uniquely extend $\delta$ to $F(t)$ by fixing the value $\delta(t)\in F(t)$. We say that $(F(t), \delta)$ is a \emph{monomial extension} of $(F, \delta)$ if $t$ is transcendental over $F$ and $\delta(t)\in F[t]$.

\begin{definition}\label{DEF:sepcial}
Let $(F(t), \delta)$ be a monomial extension of the differential field $(F, \delta)$. A polynomial $p\in F[t]$ is said to be \emph{normal} if $\gcd(p, \delta(p))=1$
and it is said to be \emph{special} if $\gcd(p, \delta(p))=p$. A rational function $f\in F(t)$ is said to be \emph{simple} if the denominator of $f$ is normal and
said to be \emph{reduced} if the denominator of $f$ is special.
\end{definition}

\begin{example}\label{EXAM:special}
Let $F = \bC$ and $t=x$. If $\delta = d/dx$, then normal polynomials in $F[t]$ are
squarefree polynomials and special polynomials are just constants in $\bC$.  If $\delta = x\cdot d/dx$, then for any $n\in \bN$, $x^n$ is a special polynomial. In fact, we can show that these are the only possible special polynomials.
\end{example}

\begin{prop}\label{PROP:val}
Let $(F(t), \delta)$ be a monomial extension of $(F, \delta)$ and $f, g \in F(t)$. For any irreducible normal polynomial $p \in F[t]$, we have
 \begin{itemize}
 \item[$(i)$] $\nu_p(fg) = \nu_p(f) + \nu_p(g)$.
 \item[$(ii)$] $\nu_p(f+g) \geq \min\{\nu_p(f), \nu_p(g)\}$ and equality holds if $\nu_p(f) \neq \nu_p(g)$.
  \item[$(iii)$] $\nu_p(\delta(f)) = \nu_p(f)-1$ if $\nu_p(f) \neq 0$,
  $\nu_p(\delta(f))\geq 0$ if  $\nu_p(f) = 0$ . In particular, for any $i\in \bN$, $\nu_p(\delta^i(f)) = \nu_p(f)-i$ if $\nu_p(f)<0$.
 \end{itemize}
\end{prop}
\begin{proof}
See Lemma 4.1.1 and Theorem 4.4.2 in~\cite[Chapter 4]{BronsteinBook}.
\end{proof}

\begin{cor}\label{COR:normal}
Let $(F(t), \delta)$ be a monomial extension of $(F, \delta)$ and $f \in F(t)$.
If $f$ is stable in $(F(t), \delta)$, then $\nu_p(f)\geq 0$ for any irreducible normal polynomial $p \in F[t]$, i.e., all of the factors of the denominator of $f$ are special polynomials.
\end{cor}
\begin{proof}
Suppose that there exists some irreducible normal polynomial $p\in F[t]$ such that $\nu_p(f)<0$. Since $f$ is stable in $(F(t), \delta)$, there exists $g_i \in F(t)$ such that $f = \delta^i(g_i)$ for any $i \in \bN$.
Let $n = -\nu_p(f)$. By Proposition~\ref{PROP:val}~$(iii)$ and the equality $f = \delta^n(g_n)$, we have $\nu_p(g_n)<0$ and
$\nu_p(\delta^n(g_n)) = \nu_p(g_n) -n < n$, which contradicts with the fact that $\nu_p(\delta^n(g_n)) = \nu_p(f) = n$.
\end{proof}
\section{Stable elementary functions}\label{SECT:elefun}
We now study the stability problem on elementary functions by looking at the classical integration problem of elementary functions through a dynamical lens.

We first recall the differential-algebraic formulation of elementary functions and their integration problems following the presentation in~\cite{BronsteinBook}.
\begin{definition}\label{DEF:elementary}
Let $(K, \delta)$ be a differential extension of $(k, \delta)$ and $t\in K$. We say that $t$ is \emph{elementary} over $k$ if one of the following conditions holds:
\begin{itemize}
\item[$(i)$] $t$ is algebraic over $k$, i.e., there exists $P\in k[X]\setminus K$ such that $P(t)=0$;
\item[$(ii)$] $t$ is exponential over $k$, i.e., $t\neq 0$ and there exists $a\in k$ such that $\delta(t)/t = \delta(a)$, where we write symbolically $t = \exp(a)$;
\item[$(iii)$] $t$ is logarithmic over $k$, i.e., there exists $a\in k\setminus\{0\}$ such that $\delta(t)=\delta(a)/a$, where we write symbolically $t = \log(a)$.
\end{itemize}
$K$ is an \emph{elementary extension} of $k$ if there are $t_1, \ldots, t_n\in K$ such that $K = k(t_1, \ldots, t_n)$ and $t_i$ is elementary over $k(t_1, \ldots, t_{i-1})$ for $i\in \{1, \ldots, n\}$. An element $f\in k$ is said to be \emph{elementary integrable} over $k$ if there exists an elementary extension
$K$ of $k$ and $g\in K$ such that $f = \delta(g)$. An elementary function is
an element of any elementary extension of the field $(\bC(x), d/dx)$.
\end{definition}


The classical problem of integration in finite terms is deciding whether
a given elementary function is elementary integrable over $\bC(x)$ or not.
For example, elementary functions
\[ \exp(x^2), \quad \frac{\exp(x)}{x},
\quad \frac{\log(x)}{x-1}, \quad \frac{1}{\sqrt{x(x-1)(x-2)}}\]
are not elementary integrable over $\bC(x)$. Liouville's theorem is the fundamental principle for elementary integration~\cite[Chap.\ 5.5]{BronsteinBook}. For later use, we recall this theorem in the case when the constant field is
algebraically closed since we will study elementary functions over $\bC(x)$.

\begin {theorem}[Liouville's theorem]\label{THM:liouville}
Let $(K, \delta)$ be a differential field with its constant subfield being algebraically closed  and $f \in K$. If there exist an elementary extension $(E, \delta)$ of $(K, \delta)$ and $g\in E$ such that $f = \delta(g)$, then
there exist $g\in K$, $h_1, \ldots, h_n\in K\setminus \{0\}$ and $c_1, \ldots, c_n\in C_K$ such that
\[f = \delta(g) +\sum_{i=1}^n c_i \frac{\delta(h_i)}{h_i}.\]
\end{theorem}

\begin{definition}\label{DEF:elemstab}
Let $(K, \delta)$ be a differential field and $f\in K$.
We say that $f$ is \emph{stable} over the elementary extensions of $K$ if
there exists a sequence $\{g_i\}_{i\in \bN}$ such that for all $i\in \bN$, $g_i$ is an element of some elementary extension of $K$ and $f = \delta^i(g_i)$.
An elementary function that is stable over the elementary extensions of $\bC(x)$ is called a \emph{stable elementary functions}.
\end{definition}

\begin{lemma}\label{LEM:logmonomial}
Let $(F, \delta)$ be  a regular differential
field with $\delta(x)=1$ and $C_F$ being algebraically closed.
Let $(F(t), \delta)$ be a monomial extension of $(F, \delta)$ with $C_{F(t)}=C_F$.
If $p\in F[t]{\setminus} F$ is a normal polynomial, then
$T = \log(p)$ is not elementary integrable over $F(t)$.
\end{lemma}
\begin{proof}
Suppose that $T$ is elementary integrable over $F(t)$. Then Theorem 5.8.2 in~\cite{BronsteinBook} implies that there are $c\in C_F$ and $a, b\in F(t)$ such that $T = \delta(cT^2 + aT) + b$ and $b$ is elementary integrable over $F(t)$.
By equating the coefficients in $T$, we get
\[1 = 2c\frac{\delta(p)}{p} + \delta(a) \quad \text{and} \quad a\frac{\delta(p)}{p} + b=0\]
Since $p$ is normal, Proposition~\ref{PROP:val}~$(iii)$ implies that $c=0$. Then $a=x+\lambda$ for some $\lambda \in C_F$ and $b = -(x+\lambda)\delta(p)/p$.
By Lemma~5.6.2 in~\cite{BronsteinBook}, we get that $x+\lambda$ must be a constant since $b$ is elementary integrable. This is a contradiction.
\end{proof}
As a special case of the above lemma, we have that $\log(\log(x))$ is not elementary integrable over $\bC(x)$.  We now extend Corollary~\ref{COR:normal} to the stability over the
elementary extensions.
\begin{theorem}\label{THM:stabmonomial}
Let $(F, \delta)$ be  a regular differential
field with $\delta(x)=1$ and $C_F$ being algebraically closed.
Let $(F(t), \delta)$ be a monomial extension of $(F, \delta)$ with $C_{F(t)}=C_F$.
If $f\in F(t)$ is stable over the elementary extensions of $F(t)$, then there exists $g\in F(t)$ such that $f - \delta(g)$ is reduced.
\end{theorem}
\begin{proof}
Suppose that there exists some irreducible normal polynomial $p\in F[t]$ such that $\nu_p(f)<0$. Then the denominator of $f$ contains normal factors. Since $f$ is stable over the elementary extensions of $F(t)$, it is also elementary integrable over $F(t)$. By the Hermite reduction and Theorem 5.6.1 in~\cite{BronsteinBook}, there exist $g, h\in F(t)$, distinct irreducible normal
polynomials $p_1, \ldots, p_n \in F[t]$ and $c_1, \ldots, c_n\in C_F$ such that
$g$ is proper with its denominator having only normal factors, $h$ is reduced and
\[f = \delta(g) + \sum_{i=1}^n c_i \frac{\delta(p_i)}{p_i}  + h.\]
It remains to show that all of the $c_i$'s are zero. Otherwise, the $\log(p_i)$'s
will be in the anti-derivatives of $f$ which are again elementary integrable
by the stability assumption on $f$. This is a contradiction with Lemma~\ref{LEM:logmonomial}.
\end{proof}

We now consider the stability problem on elementary functions.
\begin{problem}\label{PROB:elemstab}
For a given elementary function $f$ over~$\bC(x)$, decide whether $f$ is
stable or not.
\end{problem}
By Lemma~\ref{LEM:stable}, an elementary function $f$ is stable
if and only if for all $i\in \bN$, $x^i f$ is elementary integrable over $\bC(x)$.
So the stability problem can be viewed as a parametrized version of the integration
problem. For this moment, we are far from having a complete solution to the above problem.
In the rest of this section, we will study the problem on three special families of elementary functions.

\subsection{Stable rational functions}\label{SUBSECT:rat}
It is well-known that all rational functions
are elementary integrable over $(\bC(x), d/dx)$ since their anti-derivatives
are linear combinations of rational functions and logarithmic functions
over $\bC(x)$.
\begin{theorem}\label{THM:rational}
Let $f$ be a rational function in $\bC(x)$. Then
\begin{itemize}
\item[$(i)$] $f$ is a stable elementary function if $\delta=d/dx$;
\item[$(ii)$] If $\delta=d/dx$, then $f$ is stable in $(\bC(x), \delta)$ if and only if $f$ is a polynomial in $\bC[x]$;
\item[$(iii)$] If $\delta=x\cdot d/dx$, then $f$ is stable in $(\bC(x), \delta)$ if and only if $f$ is a  Laurent polynomial in $\bC[x, x^{-1}]$ that is not a nonzero constant.
\end{itemize}
\end{theorem}
\begin{proof}
$(i)$ By Lemma~\ref{LEM:stable}, $f$ is stable in the elementary extensions of $\bC(x)$ if and only if for all $i\in \bN$, $x^if$ is elementary integrable. The latter is always true for rational functions.

$(ii)$ Since the differential field $(\bC(x), d/dx)$ is regular, all polynomials are stable by the second assertion in Proposition~\ref{PROP:stable}. Suppose that $f\in \bC(x)$ is stable. By Corollary~\ref{COR:normal}, the denominator of $f$ has
only special-polynomial factors. Since in the monomial extension $(\bC(x), d/dx)$
of $(\bC, d/dx)$, the only possible special polynomials are constants in $\bC$. So
$f$ must be a polynomial in $\bC[x]$.

$(iii)$ By Theorem 5.1.2 in~\cite{BronsteinBook}, the only possible special polynomials in the monomial extension $(\bC(x), x\cdot d/dx)$ of $(\bC, x\cdot d/dx)$ are of the form $x^i$ with $i\in \bN$. If $f$ is stable in $(\bC(x), x\cdot d/dx)$, then $f$ must be a Laurent polynomial in $x$ by Corollary~\ref{COR:normal}.
Any nonzero constant $c\in \bC$ is not integrable (with respect to $x\cdot d/x$) in $\bC(x)$ since $c/x \neq d/dx(g)$ for any $g\in \bC(x)$. So $f$ is also not a nonzero constant. The necessity follows from the fact that $x^i$ is stable in $(\bC(x), x\cdot d/dx)$ for any nonzero $i\in \bZ$ since
$x^i = (x\cdot {d}/{dx})^m ({x^i}/{i^m})$ for all $m\in \bN$.
\end{proof}
\begin{remark}
Another way to show the second assertion in the above theorem is using Theorem~\ref{THM:stable}.  Suppose that $f\in \bC(x)$ is stable and it is not a polynomial. Then $f = P/Q$ for some $P, Q\in \bC[x]$ with $Q\notin \bC$ and $\gcd(P, Q) =1$. Since $Q$ is not a constant, $Q$ has at least one root in $\bC$, say $\alpha$. Let $R = Q/(x-\alpha) \in \bC[x]$. By Theorem~\ref{THM:stable}, the product $Rf = P/(x-\alpha)$ is also stable in $\bC(x)$, which leads to a contradiction since $Rf$ is not integrable in $\bC(x)$.
\end{remark}

\subsection{Stable logarithmic functions}\label{SUBSECT:log}
In calculus, we have the following  formula~\cite[p.\ 238]{GRtable2007}
\begin{equation}\label{EQ:GR}
\int x^n \log(x)^m \, dx =  \frac{x^{n+1}}{m+1} \sum_{k=0}^n (-1)^k
\frac{(m+1)!}{(m-k)!}  \frac{(\log(x))^{m-k}}{(n+1)^{k+1}}
\end{equation}
for any $m\in \bN$ and $n\in \bZ$ with $n\neq -1$. This implies that $\log(x)^m$ is a stable elementary functions. One of classical results on the integration of
logarithmic functions is as follows.

\begin{theorem}[Liouville-Hardy theorem]\label{THM:LH}
Let $f\in \bC(x)$, then $f\cdot \log(x)$ is elementary integrable over $\bC(x)$
if and only if \[f = \frac{c}{x} + \frac{dg}{dx}\] for some $c\in \bC$ and $g\in \bC(x)$.
\end{theorem}
\begin{proof}
See~\cite[p.\ 60]{Hardy1916} or~\cite{MZ1994}.
\end{proof}

Let $(K,\delta)$ be a differential field and
$t$ be a logarithmic monomial over $K$, i.e., $t$
is transcendental over $K$, $\delta(t) = \delta(a)/a$ for some $a\in K\setminus\{0\}$ and $C_{K(t)} = C_K$. Symbolically,
we write $t = \log(a)$. In $K[t]$, all special polynomials
are constants by~\cite[Theorem 5.1.1]{BronsteinBook}. So there is
no proper reduced rational function in $K(t)$. For a given $f\in K(t)$, the integration procedure
in~\cite[Chap.\ 5]{BronsteinBook} can be summarized as follows.
First, applying the Hermite reduction to $f$ yields the decomposition
\[f = \delta(g) + \frac{a}{b},\]
where $g\in K(t)$ and $a, b\in K[t]$ with $b$ being normal.
So this step reduces the integrability problem to that of simple rational
functions in $K(t)$. By the residue criterion~\cite[Theorem 5.6.1]{BronsteinBook}, there exist $c_1, \ldots, c_n \in C_K$ and normal
polynomials $b_1, \ldots, b_n\in K[t]$ such that
$f - \sum_{i=1}^n c_i \delta(b_i)/b_i \in K[t]$ if $f$ is elementary integrable over $K(t)$.  If we detect
that $f$ is not elementary integrable at this step, we can stop.
Otherwise, it remains to consider the integrability problem on polynomials in $K[t]$. For a polynomial $p\in K[t]$, there exists $q\in K[t]$ such that
$p - \delta(q) \in K$ if $p$ is elementary integrable over $K(t)$ by ~\cite[Theorem 5.8.1]{BronsteinBook}. So we either detect the non-integrability or reduce the problem from $K(t)$ to $K$. Then we proceed recursively.

In the rest of this subsection, we specialize to the case in which $K= \bC(x)$ and $\delta = d/dx$ and present a stable version of Theorem~\ref{THM:LH}.

\begin{theorem}\label{THM:ratlog}
Let $T = f\cdot \log(x)$ with $f \in \bC(x)$. Then $T$ is stable over the elementary extensions of $\bC(x)$ if and only if $f\in \bC[x, x^{-1}]$.
\end{theorem}
\begin{proof}  To show the necessity, we suppose that $T$ is stable over the elementary extension of $\bC(x)$.
By the Liouville-Hardy theorem, we have $f = c_1/x + \delta(g_1)$
for some $c_1\in \bC$ and $g_1\in \bC(x)$. Moreover,
\[ f \log(x) = \delta\left(\frac{c_1}{2} (\log(x))^2 + g_1 \log(x)\right) - \frac{g_1}{x}.\]
Since $(\log(x))^2$ and all rational functions are stable, we have $g_1\log(x)$
is also stable. Applying the Liouville-Hardy theorem to $g_1\log(x)$ yields
$g_1 = c_2/x + \delta(g_2)$ for some $c_2\in \bC$ and $g_2\in \bC(x)$. Iterating this process, we obtain two sequences $\{c_i\}_{i\in \bN}$ in $\bC$ and $\{g_i\}_{i\in \bN}$ in $\bC(x)$ with $g_0 =  f$ and
$g_i = c_{i+1}/x + \delta(g_{i+1})$ for all $i\in \bN$. If the denominator of $f$
has a root other than zero, so do the $g_i$'s. Choosing sufficiently large $i$ yields a contradiction by looking at the order at this root. So $f$ is a Laurent polynomial in $x$. For the necessity, we only need to show that $x^m \log(x)$
is stable for any $m\in \bZ$. By the formula~\eqref{EQ:GR}, it suffices to show that
$\log(x)/x$ is elementary integrable. This is true since $\log(x)/x = \delta(\log(x)^2/2)$.
\end{proof}

\subsection{Stable exponential functions}\label{SUBSECT:exp}
Throughout this part, let $\delta$ be the usual derivation $d/dx$
on $\bC(x)$ and its extensions. The non-elementary integrability of
$\exp(x^2)$ is derived from the
following theorem by Liouville in~\cite{Liouville1835} (see~\cite[p.\ 971]{Rosenlicht1972} for its proof).
\begin{theorem}\label{THM:fexpg}
Let $f, g \in \bC(x)$ with $g \notin \bC$ and $t = f\cdot \exp(g)$. Then $t$ is elementary integrable over $\bC(x)$ if and only if there exists $h\in \bC(x)$
such that $f = \delta(h) + h \cdot \delta(g)$.
\end{theorem}

A rational function $f = a/b$ with $a, b\in \bC[x]$ and $\gcd(a, b)=1$ is said to be \emph{differential-reduced} (with respect to $\delta$) if
\[\gcd(b, a-i\delta(b)) = 1 \quad \text{for all $i\in \bZ$}.\]
We recall some  basic properties of differential-reduced rational functions from~\cite{GeddesLeLi2004, BCCLX2013}.

\begin{prop}\label{PROP:dred}
Let $f = a/b\in \bC(x)$ be such that $a, b\in \bC[x]$ and $\gcd(a, b)=1$. Then
\begin{itemize}
\item[$(i)$]  $\delta(f)$ is always differential-reduced;
\item[$(ii)$] $f + m\delta(b)/b$ is differential-reduced for any $m\in \bZ$
if $f$ is differential-reduced;
\item[$(iii)$] if $f$ is differential-reduced and $b\delta(g) + a g \in \bC[x]$, then $g \in \bC[x]$.
\end{itemize}
\end{prop}
\begin{proof} The first two assertions follows from Lemma 2 in~\cite{GeddesLeLi2004} saying that $f$ is differential-reduced if and only if none of its residues at simple poles is an integer. The third one is Lemma 6 in~\cite{BCCLX2013}.
\end{proof}

\begin{lemma}\label{LEM:expg}
Let $g\in \bC(x)$ and $f = \delta(g) = a/b$ with $a, b \in \bC[x]$ and $\gcd(a, b)=1$. Then  $t = P\cdot b^m \cdot \exp(g)$ with $m\in \bN$ and $P\in \bC[x]\setminus \{0\}$ is elementary integrable over $\bC(x)$ if and only if there exists $Q \in \bC[x]$ such that $t = \delta(Q\cdot b^{m+1}\cdot\exp(g))$ and
\[
\deg_x(Q) = \left\{
\begin{array}{ll}
\deg_x(P)- \deg_x(a), & \hbox{ if $\deg_x(a)\geq \deg_x(b)$;}\\
\deg_x(P)-\deg_x(b)+1, & \hbox{if $\deg_x(a) < \deg_x(b)-1$.}
\end{array}
  \right.
\]
\end{lemma}
\begin{proof}
Suppose that $t$ is elementary integrable over $\bC(x)$. Then Theorem~\ref{THM:fexpg} implies that there exists $h\in \bC(x)$ such that
\[P b^m = \delta(h) + f \cdot h.\]
Write $h = Q \cdot b^{m+1}$ with $Q\in \bC(x)$. Then
\[P = b\delta(Q) + (a+(m+1)\delta(b))Q.\]
We first show that $Q\in \bC[x]$.
Since $f = \delta(g) = a/b$, Proposition~\ref{PROP:dred}~$(i)$ and $(ii)$ implies that both $f$ and
$(a+(m+1)\delta(b))/b$ are differential-reduced. Then we have $Q\in \bC[x]$ by
Proposition~\ref{PROP:dred}~$(iii)$. We next estimate $\deg_x(Q)$.
 Since $f = \delta(g)$, we have either $\deg_x(a)\geq \deg_x(b)$ or $\deg_x(a) < \deg_x(b)-1$
by Theorem~4.4.4 in~\cite{BronsteinBook}, i.e., $\deg_x(a) \neq  \deg_x(b)-1$.
If $\deg_x(a)\geq \deg_x(b)$, then $\deg_x(P) = \deg_x(a) + \deg_x(Q)$. Hence
$\deg_x(Q) = \deg_x(P)- \deg_x(a)$. If $\deg_x(a) < \deg_x(b)-1$, then
the leading monomial of $b\delta(Q) + (a+(m+1)\delta(b))Q$ is $\text{lc}(Q)\cdot \text{lc}(b) \cdot(\deg_x(Q)+ (m+1)\deg_x(b))x^{\deg_x(b)+ \deg_x(Q)-1}$, which implies that $\deg_x(Q) = \deg_x(P)-\deg_x(b)+1$.
\end{proof}

We now present a stable version of Theorem~\ref{THM:fexpg}.

\begin{theorem}\label{THM:expstab}
Let $f, g \in \bC(x)$ be such that $g \notin \bC$ and $t = f\cdot \exp(g)$. Then $t$ is a stable
elementary function if and only if $f\in \bC[x]$ and $g = \lambda x +\mu$
for some $\lambda, \mu \in \bC$ with $\lambda\neq 0$.
\end{theorem}
\begin{proof}To see the necessity, note that the function
$\exp(\lambda x + \mu)$ is stable since for any $i\in \bN$,
\[\exp(\lambda x + \mu) = \delta^i\left(\frac{1}{\lambda^i}\exp(\lambda x + \mu)\right).\]
By Theorem~\ref{THM:stable}, $f\cdot \exp(\lambda x + \mu)$ is stable for any $f\in \bC[x]$. For the sufficiency, we assume that $t$ is a stable
elementary function. Since $g\notin \bC$, we have $\delta(g)\neq 0$. We first claim that $\delta(g)$ must be a constant.
Suppose that this is not true. Let $d$ be the denominator of $f$. Then $t_1 = d\cdot t = p \cdot \exp(g)$ for some $p\in \bC[x]$ and it is also stable by Theorem~\ref{THM:stable}.
Since $t_1$ is elementary integrable over $\bC(x)$, Theorem~\ref{THM:fexpg} implies that there exists $h\in \bC(x)$ such that
\[p = \delta(h) + \delta(g) \cdot h.\]
Write $\delta(g) = a/b$ with $a, b \in \bC[x]$ and $\gcd(a, b)=1$.
By Lemma~\ref{LEM:expg}, there exists  $Q_1\in \bC[x]$ such that $h = b Q_1$ and
\[
\deg_x(Q_1) = \left\{
\begin{array}{ll}
\deg_x(p)- \deg_x(a), & \hbox{ if $\deg_x(a)\geq \deg_x(b)$;}\\
\deg_x(p)-\deg_x(b)+1, & \hbox{if $\deg_x(a) < \deg_x(b)-1$.}
\end{array}
  \right.
\]
Notice that $\deg_x(Q_1)< \deg_x(p)$ since $\delta(g)$ is not a constant.
Since $t_1$ is stable and $t_1 = \delta(Q_1 b\exp(g))$, we have $t_2 = Q_1 b\exp(g)$ is also stable. Following the same argument as above, we have
 $t_2 = \delta(Q_2\cdot  b^2\exp(g))$ for some $Q_2 \in \bC[x]$ with $\deg_x(Q_2)< \deg_x(Q_1)$. Repeating this process, we obtain a sequence of polynomials
 $\{Q_i\}_{i\in \bN}$ with $Q_0 = p$ and $\deg_x(Q_{i+1})< \deg_x(Q_i)$. This is a contradiction. Thus, $\delta(g)$ is a nonzero constant and then $g = \lambda x +\mu$ for some $\lambda, \mu \in \bC$ with $\lambda\neq 0$. It remains to show that $f \in \bC[x]$. Suppose that $f\in \bC(x)$ not a polynomial. Then $f = P/Q$ for some $P, Q\in \bC[x]$ with $Q\notin \bC$ and $\gcd(P, Q) =1$. Since $Q$ is not a constant, $Q$ has at least one root in $\bC$, say $\alpha$. Let $R = Q/(x-\alpha) \in \bC[x]$. Since $f\exp(g)$ is stable, so is $R\cdot f \exp(g) = P\exp(g)/(x-\alpha)$ by Theorem~\ref{THM:stable}. Now again Theorem~\ref{THM:fexpg} implies that there exists $h\in \bC(x)$ such that
\[\frac{P}{x-\alpha} = \delta(h) + \lambda \cdot h.\]
By estimating the valuation at $x-\alpha$, we have $\nu_{x-\alpha}(P/(x-\alpha))=-1$ but $\nu_{x-\alpha}(\delta(h) + \lambda \cdot h)$ is either
non-negative or strictly less than $-1$, which leads to a contradiction.
\end{proof}

\section{Stable D-finite power series}\label{SECT:dfinite}
The notion of D-finite power series was first introduced by Stanley~\cite{Stanley1980} in 1980 and studied extensively in~\cite{Lipshitz1989, Zeilberger1990, EC2}. These series like algebraic numbers can be algorithmically manipulated via its defining linear differential equations~\cite{AbramovLeLi2005, Salvy2019}.  We will study
the stability problem on D-finite power series.

Let $k$ be a field of characteristic zero and let $k{[[}x{]]}$ be the ring of formal power series in $x$ over $k$ and let $\delta$ denote the derivation
with $\delta(x)=1$ and $\delta(a)=0$ for all $a\in k$.
The quotient field of $k{[[}x{]]}$ is called
the field of formal Laurent series, denoted by $k((x))$, which includes
$k(x)$ as its subfiled.  Let $R$ be a ring and $\sigma: R\rightarrow R$
be an isomorphism on $R$. We call the pair $(R, \sigma)$
a \emph{differential ring}. Over a difference ring, we have the polynomial ring $R\langle S \rangle$ in which the addition is defined coefficient-wise
in the indeterminate $S$ and the multiplication satisfies the rule:
\[S \cdot r= \sigma(r)\cdot S\quad  \text{for all $r\in R$}.\]
In particular, we call $k[x]\langle S \rangle$ the ring of linear recurrence operators with polynomial coefficients. For a sequence $a_n: \bN \rightarrow k$, we define an action of an operator $P = \sum_{i=0}^d p_i S^i \in k[x]\langle S \rangle$ on $a_n$ by $P(a_n) = \sum_{i=0}^d p_i(n) a_{n+i}$. For an operator $L = \sum_{i=0}^d \ell_i D^i\in k[x]\langle D \rangle$ with $\ell_d\neq 0$, we call $d$ the order of $L$, denoted by $\text{ord}(L)$ and $\max\{\deg_x(\ell_0), \ldots, \deg_x(\ell_d)\}$ the degree of $L$, denoted by $\deg(L)$. We define the order and the degree of an operator $P \in k[x]\langle S \rangle$ in a similar way.

\begin{definition}\label{DEF:dfinit}
A power series $f = \sum_{n\geq 0} a_n x^n \in k{[[}x{]]}$ is said to be \emph{D-finite} over $k(x)$ if there exists a non-zero operator $L \in k[x]\langle D \rangle$ such that $L(f)=0$. Such an operator $L$ is called an \emph{annihilator} for $f$.
A sequence $a_n: \bN \rightarrow K$
is said to be \emph{P-recursive}  if there exists a non-zero operator
$P\in \in k[x]\langle S \rangle$ such that $P(a_n)=0$. We also call such an operator $P$ an \emph{annihilator} for $a_n$.
\end{definition}

The following theorem summarizes some fundamental properties of
D-finite power series and P-recursive sequences.

\begin{theorem}\label{THM:dfinite}
Let $f = \sum_{n\geq 0} a_n x^n \in k{[[}x{]]}$. Then
\begin{itemize}
\item[$(i)$] $f$ is D-finite if and only if its coefficient
sequence $a_n$ is P-recursive;
\item[$(ii)$] if $a_n$ has an annihilator in $k[x]\langle S\rangle$
of order $r$ and degree $d$, then $f$ has an annihilator in $k[x]\langle D\rangle$ of order at most $d$ and degree at most $r+d$;
\item[$(iii)$] if both $a_n$ and $b_n$ are P-recursive with annihilators $A$ and $B$ in $k[x]\langle S\rangle$, then $a_n b_n$ is also P-recursive with
an annihilator of order at most $\text{ord}(A)\text{ord}(B)$ and degree
at most \[2\max\{\deg(A), \deg(B)\} \text{ord}(A)^2\text{ord}(B)^2.\]
\end{itemize}
\end{theorem}
\begin{proof}
For the proofs, see~\cite[p.\ 149]{KPbook} and~\cite[Theorem 8]{Kauers2014}.
\end{proof}
As a high-order generalization of Gosper's algorithm~\cite{Gosper1978}
and its differential analogue~\cite{Almkvist1990}, Abramov and van Hoeij studied the integration problem on solutions of linear functional
equations~\cite{AvH1997, AvH1999}. In the following, we let $k =\bC$ and
$\delta = d/dx$ so that we can talk about generalized series solutions of linear differential equations.
\begin{problem}\label{PROB:dfiniteint}
Given an operator $L \in k(x)\langle D \rangle$, find a minimal-order operator $\tilde{L} \in  k(x)\langle D \rangle$  such that
the derivatives of the solutions of $\tilde{L}$ are the solutions of $L$.
We call $\tilde{L}$ an \emph{integral} of $L$, denoted by $\text{int}(L)$.
\end{problem}
Note that $\tilde L$ is unique up to a factor in $K$ and
the order of $\tilde L$ is greater than that of $L$ by at most one.
If $\text{ord}(\tilde{L})= \text{ord}(L)$,  Abramov and van Hoeij
proved that there exist $P \in k(x)\langle D \rangle$ with $\text{ord}(P)<\text{ord}(L)$ and $r\in k(x)$ such that
$D\cdot P + r\cdot L = 1$. In this case,
a solution $f$ of $L$ has an anti-dervative of the form $P(f)$.
For a power series $f = \sum_{n\geq 0} a_nx^n \in k{[[}x{]]}$, we call the series $\sum_{n\geq 1} \frac{a_{n-1}}{n} x^n$ a \emph{formal integral} of $f$, denoted by $\text{int}(f)$.

\begin{definition}\label{DEF:dstable}
 A D-finite power series $f$ is said to be \emph{stable} if
there exists a sequence $\{g_i\}_{i\in \bN}$ in $k{[[}x{]]}$ such that $g_0 =f$, $g_i = \delta(g_{i+1})$ and  all of the $g_i$'s have annihilators of the same order.  It is said to be \emph{eventually stable} if there exists $m\in \bN$ such that $\text{int}^m{(f)}$ is stable.
An operator $L \in k(x)\langle D \rangle$ is said to be~\emph{stable}
if $\text{ord}(\text{int}^i{(L)})= \text{ord}(L)$ for all $i \in \bN$
and be \emph{eventually stable} if there exists $m\in \bN$ such that $\text{int}^m{(L)}$ is stable.
\end{definition}

The following result was first discovered by Guo and then proved by the author in 2020. For more interesting stable power series, see Guo's thesis~\cite{Guo2020}.
\begin{theorem}\label{THM:dstable}
Any D-finite power series is eventually stable.
\end{theorem}
\begin{proof}
Let $f = \sum_{n\geq 0} a_n x^n \in k{[[}x{]]}$ be D-finite and $L\in k[x]\langle D \rangle$ be the minimal annihilator for $f$. By Theorem~\ref{THM:dfinite}~$(i)$, $a_n$ is P-recursive and so has a minimal annihilator $P\in k(x)\langle S\rangle$.
Let
\[g_i = \sum_{n\geq i} \frac{a_{n-i}}{n(n-1)\cdots (n-i+1)} x^n.  \]
Then $f = \delta^i(g_i)$. We claim that the order of minimal annihilators for $g_i$ is bounded. By Theorem~\ref{THM:dfinite}~$(ii)$, it suffices to show that the degrees of the minimal annihilators of the coefficient sequences $b_{n, i} = a_{n-i}/(n(n-1)\ldots (n-i+1))$  are bounded. Note that $a_{n-i}$ has the same annihilator as
$a_n$ and the minimal annihilators of the sequences $t_i = 1/(n(n-1)\ldots (n-i+1))$ are $(n+1)S -(n-i+1)$, whose order and degree are independent of $i$. Then Theorem~\ref{THM:dfinite}~$(iii)$ implies that the degrees of the minimal annihilators for $b_{n ,i}$ are at most $2\max\{1, \deg(P)\}\text{ord}(P)^2$. Then there exists $m\in \bN$ such that the formal integrals $\text{int}^m(f)$ is stable.
 \end{proof}

\section{Conclusion and future work}\label{SECT:conc}
This paper presents some
initial results towards a deep connection between dynamics and differential
algebra with the focus on stability problems in symbolic integration.
This is just a first try and more general cases are waiting for further
studying in this direction.

To conclude this paper, we propose some problems for future work.
The first problem is characterizing all possible algebraic functions that are stable in the differential field $(\overline{\bC(x)}, d/dx)$. The typical stable family of algebraic functions is $(x-c)^r$ with $c\in \bC$ and
$r \in \bQ \setminus\{-1, -2, \ldots\}$. We conjecture that an algebraic function is stable in $(\overline{\bC(x)}, d/dx)$ if and only if it is of the form
\[\sum_{i=1}^n p_i \cdot (x-c_i)^{r_i},\]
where $p_i \in \bC[x]$, $c_i \in \bC$ and $r_i \in \bQ \setminus\{-1, -2, \ldots\}$. The second problem is formulating a stable version of Liouville's theorem that describes the structure of elementary functions that are elementary integrable. This will be crucial for developing a Risch-type algorithm for detecting whether an elementary function is stable or not.
The third problem is studying stability problems in symbolic summation. For this moment, we have some parallel results in this direction which will be included in a forthcoming paper. A special case of this problem is to characterize all possible stable hypergeometric terms with respect to the
difference operator by thinking Gosper's algorithm dynamically. The last problem is related to a classical open problem  in differential algebra
asked by Rubel in~\cite{Rubel1983}. For a D-finite power series
$f = \sum_{n\geq 0} a_n x^n\in k[[x]]$, Rubel conjectured that
the set $\{n\in \bN\mid a_n = 0 \}$ is a Skolem set, i.e., a union of finitely many arithmetic progressions. A given elementary function $f(x)$ over $\bC(x)$ is generically not stable, but we conjectured that the set $\{i\in \bN \mid \text{$x^if(x)$ is elementary integrable over $\bC(x)$} \}$ is also a Skolem set. For instance, $\exp(x^2)$ is not elementary integrable, but for odd $i\in \bN$, $x^i \exp(x^2)$ is elementary integrable.

\medskip
\noindent {\bf Acknowledgement.}
The author thanks Ruyong Feng and Umberto Zannier for many discussions during the formation of
the conjecture for stable algebraic functions and  Ze-wang Guo for sharing his discoveries and many interesting examples.  I am also very grateful to  Hao Du and all of my PhD students (Lixin Du, Pingchuan Ma, and Chaochao Zhu) for reading the draft and sending me their constructive comments.

\bibliographystyle{plain}

\end{document}